\documentclass[conference,letterpaper]{IEEEtran}

\usepackage[utf8]{inputenc} 
\usepackage[T1]{fontenc}
\usepackage{url}
\usepackage{ifthen}
\usepackage[cmex10]{amsmath}
\usepackage{amsmath,bm}
\usepackage{amsthm}
\usepackage{bbm}
\usepackage{amssymb}
\usepackage{float}
\usepackage{mathtools}
\usepackage{thmtools}
\usepackage{hyperref}
\usepackage{graphicx}
\usepackage{xcolor}
\usepackage[english]{babel}
\usepackage{tabularx}
\usepackage{multirow}
\usepackage{stfloats}
\usepackage[normalem]{ulem}
\usepackage{cite}
\usepackage[ruled,vlined]{algorithm2e}
\SetKwInput{KwInput}{Input}                
\SetKwInput{KwOutput}{Output}              
\usepackage{algorithmic}
\usepackage{todonotes}
\usepackage{mleftright}\mleftright

\usepackage{graphicx}

\usepackage{tikz}\usetikzlibrary{shapes.multipart}
\usetikzlibrary{positioning}
\tikzset{font={\fontsize{9pt}{12}\selectfont}}
\tikzset{>=latex}
\usetikzlibrary{calc}
\usetikzlibrary{arrows}
\usetikzlibrary{external}
\usetikzlibrary{patterns}
\usetikzlibrary{fit}
\usepackage{changes}

\DeclareMathOperator{\wt}{wt}

\newtheorem{theorem}{Theorem$\!$}
\newtheorem{lemma}{Lemma$\!$}
\newtheorem{claim}{Claim$\!$}
\newtheorem{corollary}{Corollary$\!$}
\newtheorem{proposition}{Proposition$\!$}
\theoremstyle{definition}
\newtheorem{construction}{Construction$\!$}
\newtheorem{definition}{Definition$\!$}
\newtheorem{example}{Example$\!$}

\usepackage{graphicx}

\usepackage{xcolor}

\newcommand{\cB}{\mathcal{B}}
\newcommand{\cC}{\mathcal{C}}

\newcommand{\cM}{\mathcal{M}}

\newcommand{\mybold}[1]{\bm{#1}}
\newcommand{\ba}{{\mybold{a}}}

\newcommand{\bb}{{\mybold{b}}}
\newcommand{\bc}{{\mybold{c}}}
\newcommand{\bd}{{\mybold{d}}}
\newcommand{\be}{{\mybold{e}}}

\newcommand{\bs}{{\mybold{s}}}

\newcommand{\bu}{{\mybold{u}}}
\newcommand{\bv}{{\mybold{v}}}

\newcommand{\bx}{{\mybold{x}}}

\newcommand{\by}{{\mybold{y}}}
\newcommand{\bz}{{\mybold{z}}}

\newcommand{\bgamma}	{\mybold{\gamma}}

\newcommand{\bepsilon}	{\mybold{\epsilon}}

\newcommand{\bphi}		{\mybold{\phi}}

\newcommand{\bpsi}		{\mybold{\psi}}

\newcommand{\defeq}{\mathrel{\stackrel{\makebox[0pt]{\mbox{\normalfont\tiny def}}}{=}}}

\usepackage{calligra} \DeclareMathAlphabet{\mathcalligra}{T1}{calligra}{m}{n} \DeclareFontShape{T1}{calligra}{m}{n}{<->s*[2.2]callig15}{}

\begin{document}

\title{Codes for Correcting $t$ Limited-Magnitude \\Sticky Deletions}

\author{\textbf{Shuche Wang}\IEEEauthorrefmark{1},
        \textbf{Van Khu Vu}\IEEEauthorrefmark{4},
and        \textbf{Vincent Y.~F.~Tan}\IEEEauthorrefmark{2}\IEEEauthorrefmark{3}\IEEEauthorrefmark{1}
        \\[0.5mm]
\IEEEauthorblockA{
\IEEEauthorrefmark{1} \small Institute of Operations Research and Analytics, National University of Singapore, Singapore \\[0.5mm]
\IEEEauthorrefmark{2} \small Department of Mathematics, National University of Singapore, Singapore\\[0.5mm]
\IEEEauthorrefmark{3} \small Department of Electrical and Computer Engineering, National University of Singapore, Singapore\\[0.5mm]
\IEEEauthorrefmark{4} \small Department of Industrial Systems Engineering
and Management, National University of Singapore, Singapore\\[0.5mm]
}
{Emails:\, shuche.wang@u.nus.edu,   isevvk@nus.edu.sg,  vtan@nus.edu.sg }      }

\maketitle
\begin{abstract}
    Codes for correcting sticky insertions/deletions and limited-magnitude errors have attracted significant attention due to their applications of flash memories, racetrack memories, and DNA data storage systems. In this paper, we first consider the error type of $t$-sticky deletions with $\ell$-limited-magnitude and propose a non-systematic code for correcting this type of error with redundancy $2t(1-1/p)\cdot\log(n+1)+O(1)$, where $p$ is the smallest prime larger than $\ell+1$. Next, we present a systematic code construction with an efficient encoding and decoding algorithm with redundancy  $\frac{\lceil2t(1-1/p)\rceil\cdot\lceil\log p\rceil}{\log p} \log(n+1)+O(\log\log n)$, where $p$ is the smallest prime larger than $\ell+1$.
\end{abstract}

\section{introduction}

Coding techniques for data storage technologies, such as flash memories, racetrack memories, and DNA data storage systems, have attracted significant interest recently. Unlike conventional communication and storage systems, substitutions are not the dominant type of error in these emerging storage systems. \textit{Sticky insertions and deletions} are prevalent among the file synchronization, racetrack memories~\cite{chee2018coding} and DNA data storage~\cite{jain2017duplication,kovavcevic2018asymptotically}. In addition, \textit{limited-magnitude errors} occur frequently in flash memories where information is stored in the corresponding level of cells \cite{cassuto2010codes}. Also, part of the information is represented by the lengths of runs in some DNA data storage systems \cite{lee2019terminator} and \textit{limited-magnitude errors} can occur in the process of synthesizing DNA sequences.

The problem of correcting sticky deletions/insertions was first studied by Levenshtein, who proposed a construction for correcting $r$ sticky-insertions/deletions in \cite{levenshtein1965binary}. Levenshtein also provided a lower bound and an upper bound on the largest size of the code for correcting $r$ sticky-insertions/deletions\cite{levenshtein1965binary}. Dolecek et al.\  \cite{dolecek2010repetition} proposed a code for correcting $r$ sticky-insertions with size 
at most $2^{n+r}/n^r$, which improves the lower bound in \cite{levenshtein1965binary}. In addition to non-efficient code construction, Mahdavifar et al.~\cite{mahdavifar2017asymptotically} proposed an asymptotically optimal systematic sticky-insertion-correcting code. Besides, sticky deletions/insertions and duplication deletions can be considered as asymmetric deletions/insertions via the Gray mapping~\cite{tallini2010efficient}. Tallini et al.~\cite{tallini2008new,tallini2010efficient,tallini20111,tallini20131,tallini2022deletions} provided a series of theories and code designs for correcting asymmetric deletions/insertions.
For the sticky-deletions/insertions code in the practical storage system, Chee et al.~\cite{chee2018coding}  presented constructions of codes for sticky insertions in the racetrack memory scheme and Jain et al.~\cite{jain2017duplication} proposed codes for correcting duplication errors in DNA data storage systems, which is highly relevant to the sticky-insertion/deletion problem.   

Cassuto et al.~\cite{cassuto2010codes} studied \textit{asymmetric limited-magnitude} errors and proposed a code for correcting these errors. This type of error can be generalized as the error ball $\cB(n,t,k_{+},k_{-})$, where at most $t$ entries increase by at most $k_{+}$ and decrease by at most $k_{-}$ for a sequence with length $n$. Hence, the code for correcting $t$ asymmetric limited-magnitude errors is equivalent to a packing of $\Sigma^{n}$ by the error ball $\cB(n,t,k_{+},k_{-})$. There is a series of works studying the packing/tiling by $\cB(n,t,k_{+},k_{-})$ beginning with $t=1$ in \cite{schwartz2011quasi} and extending to the general constant $t\ge 2$ in \cite{wei2021lattice,wei2022tilings}.

In this work, our goal is to construct codes for correcting $t$-sticky-deletions with $\ell$-limited-magnitude where both $t$ and $\ell$ are constants. This means that deletions occur in at most $t$ runs and at most $\ell$ repeated bits are allowed to be deleted in each run. Our main results are the following:
\begin{itemize}
    \item We present a non-systematic code for correcting $t$-sticky-deletions with $\ell$-limited-magnitude that has redundancy $2t(1-1/p)\cdot\log(n+1)+O(1)$, where $p$ is the smallest prime larger than $\ell+1$. The redundancy of the code can be further reduced for the special case when $t=1$.
    \item We propose a systematic code for correcting $t$-sticky-deletions with $\ell$-limited-magnitude with efficient encoding and decoding algorithm that has redundancy $\frac{\lceil2t(1-1/p)\rceil\cdot\lceil\log p\rceil}{\log p} \log(n+1)+O(\log\log n)$, where $p$ is the smallest prime larger than $\ell+1$.
\end{itemize}

If we ignore the difference between  $\lceil\log p\rceil$ and $\log p$, we notice that the redundancy of our systematic code construction is only off from the non-systematic code by at most $O(\log\log n)$. 

The paper is organized as follows.  Notation and preliminaries are stated in Section~\ref{sec:notation}.  Section~\ref{sec:nonsystem} presents a non-systematic code for correcting $t$-sticky-deletions with $\ell$-limited-magnitude and provides a lower bound of the size of this code. Section~\ref{sec:system} presents efficient encoding and decoding algorithms with an analysis of redundancy and time complexity. Finally, Section~\ref{sec:conclu} concludes this paper.

\section{Notation and Preliminaries}\label{sec:notation}
We now describe the notations used throughout this paper. Let $\Sigma_q$ denote a finite alphabet of size $q$ and $\Sigma_q^n$ represent the set of all sequences of length $n$ over $\Sigma_q$. Without loss of generality, we assume $\Sigma_q=\{0,1,\dotsc,q-1\}$. For ease of notation, we will denote 
the set $\{1,2,\ldots,m\}$ as $[m]$. 
For two integers $i<j$, let $[i,j]$ denote the set $\{i,i+1,i+2, \ldots, j\}$.

We write sequences with bold letters, such as $\bx$, and their elements with plain letters, e.g., $\bx=x_1\dotsm x_n$ for $\bx\in \Sigma_q^n$. For functions, if the output is a sequence, we also write them with bold letters, such as $\bphi(\bx)$. The $i$th element in $\bphi(\bx)$ is denoted $\phi(\bx)_i$. 
$\bx_{[i,j]}$ denotes the substring beginning  at index $i$ and ending at index $j$, inclusive. A run is a maximal substring consisting of identical symbols and $n_r(\bx)$ denotes the number of runs of the sequence $\bx$. The weight $\wt(\bx)$ of a sequence $\bx$ represents the number of non-zero symbols in it. 

\begin{definition}
Define function $\bpsi: \Sigma_2^n\rightarrow \Sigma_2^n$ such that
\begin{equation*}\label{eq:psi}
    \psi(\bx)_i=\left\{
    \begin{array}{ll}
    x_{i}\oplus x_{i+1}, &  i=1,2,\dotsc,n-1\\
    x_n, & i=n 
\end{array}
\right.
\end{equation*}
where $a\oplus b$ denotes $(a+b)\bmod 2$.
\end{definition}

For a binary sequence $\bx\in\Sigma_2^n$, we can uniquely write it as $\bx=0^{r_1}10^{r_2}10^{r_3}\dotsc 10^{r_{w+1}}$, where $w=\wt(\bx)$. For the sake of convenience in the following paper, we append a bit $1$ at the end of $\bx$ and denote it as $\bx1$. Since the sequence $\bx1$ always ends with $1$, $\bx1$ can be always written as $\bx1=0^{r_1}10^{r_2}10^{r_3}\dotsc 0^{r_{w}}1$, where $w=\wt(\bx1)$. 
\begin{definition}
    Define function $\bgamma:\Sigma_2^n\rightarrow\Sigma^{w}$ and $\bgamma(\bx1)\defeq(r_1,r_2,r_3,\dotsc,r_{w})\in\Sigma^{w}$, where $\bx1=0^{r_1}10^{r_2}10^{r_3}\dotsc 0^{r_{w}}1$ and $w=\wt(\bx1)$.
\end{definition}

\begin{definition}
    Define function $\bphi:\Sigma_2^{n}\rightarrow \Sigma^{n_r}$ such that $\bphi(\bx1)=\bgamma(\bpsi(\bx1))$, where $n_r=n_r(\bx1)$. 
\end{definition}
Since both mapping functions $\bpsi$ and $\bgamma$ are one-to-one mapping functions, the mapping function $\bphi$ is also one-to-one mapping.

\begin{example}
    Suppose $\bx = 0111010100$. Then, $\bx1 = 01110101001$ with  $n_r(\bx1)=8$  and $\bpsi(\bx1) =10011111011$. Also, $\bphi(\bx1)=\bgamma(\bpsi(\bx1))=02000010$ with length 8.
\end{example}

\begin{definition}
    Given a sequence $\by\in\Sigma_q^n$, define $\partial:\Sigma_q^n\rightarrow\Sigma_2^n$ such that
\begin{equation*}\label{eq:psi}
    \partial(\by)_i=\left\{
    \begin{array}{ll}
    0, &  {\mathrm{when}}\; y_i= 0;\\
    1, & {\mathrm{otherwise}}.
\end{array}
\right.
\end{equation*}
\end{definition}
 In addition, for a sequence $\by\in\Sigma_q^n$, denote $(\by \bmod a)=(y_1 \bmod a, y_2 \bmod a, \dotsc, y_n \bmod a)$, where $a<q$.

\begin{definition}
    A sticky deletion denotes deleting $k$ repetition bits in a run of a sequence, but cannot delete the whole run.
\end{definition}

\begin{proposition}\label{clm:asticky}
    For a binary sequence $\bx\in\Sigma_2^n$, deleting $k$ repetition bits in a run of $\bx$ is equivalent to the corresponding entry of $\bphi(\bx)$ suffers a decrease by $k$.
\end{proposition}
\begin{proof}
    Based on the definition of the mapping function $\bphi$, the value of each symbol $\bphi(\bx)_i$ is the length of $i$-th run of the sequence $\bx$ minus 1. Deleting $k$ repetition bits in a run of $\bx$ means the length of this run is decreased by $k$.
\end{proof}

Therefore, \emph{$t$ sticky deletions} pattern $\bd$ is a sequence $(d_1,d_2,\dotsc,d_{n_r})$ with $\wt(\partial(\bd))=t$. Suppose $\bx\in\Sigma_2^{n}$ is transmitted and corrupted by $\bd$, the received sequence $\bx'$ should be $\bphi(\bx'1)=\bphi(\bx1)-\bd$.

Given a sequence $\by\in\Sigma_q^n$, \emph{asymmetric $\ell$-limited-magnitude $t$ errors} denote at most $t$ of entries of $\by$ suffer a decrease/an increase by as most $\ell$, the corrupted sequence can be written as $\by'=\by-\be$, where $\wt(\partial(\be))=t$ and $e_i\leq \ell, \forall[n]$. Therefore, throughout this paper, we provide the definition of \emph{$t$ sticky deletions with $\ell$-limited-magnitude}.

\begin{definition}
\emph{$t$ sticky deletions with $\ell$-limited-magnitude} denote that given a sequence $\bx\in\Sigma_2^n$, at most $t$ of entries of $\bphi(\bx1)=(u_1,\dotsc,u_{n_r})$ suffer a decrease by at most $\ell$.  The corrupted sequence is $\bx'$ and $\bphi(\bx'1)$ can be written as $\bphi(\bx'1)=(v_1,\dotsc,v_{n_r})$ with $n_r=n_r(\bx'1)=n_r(\bx1)$, where
\begin{enumerate}
    \item \emph{Limited-magnitude deletions:} $u_i- v_i\leq \ell$ and $v_i\leq u_i, \forall i\in [n_r]$;
    \item \emph{$t$ Sticky-deletions:} Number of index $i$ is at most $t$ such that $v_i\neq u_i, \forall i\in[n_r]$.
\end{enumerate}
\end{definition}

\begin{example}
    Suppose we have $\bx=0100111001\in\Sigma_2^{10}$, hence $\bphi(\bx1)=001211$. If the retrieved sequence $\bx'=010101\in\Sigma_2^{5}$ and the corresponding $\bphi(\bx'1)=000001$, by comparing $\bphi(\bx1)$ and $\bphi(\bx'1)$, we can see the limited magnitude of the deletion is $\ell\leq 2$ and the number of sticky-deletions is $t=3$.
\end{example}

Denote $\Phi$ be the set of mapping $\Sigma_2^n$ by the function $\bphi$ and $\Sigma_2^n$ as the set containing all binary sequences with length $n$. Then, we will show the cardinality of $\Phi$. The proof is given in Appendix~\ref{app:size_phix}. 
\begin{restatable}{lemma}{sizephi}\label{lem:size_phix}
     The cardinality of $\Phi$ is:
    \begin{equation*}
    |\Phi|=\sum_{n_r=1}^{n+1}{n \choose n_r-1}=2^n.
    \end{equation*}
\end{restatable}
For shorthand, let $D_{t,\ell}(\bx) \subseteq \Sigma_2^{n-t\ell}$ denote the set of all sequences possible given that $t$ sticky-deletions with $\ell$-limited-magnitude occur to $\bx$. The size of a code $\cC\subseteq \Sigma_2^n$ is denoted $|\cC|$ and its redundancy is defined as $\log (2^n/|\cC|)$, where all logarithms without a base in this paper are to the base 2. We say that a code $\cC_{t,\ell} \subseteq \Sigma_2^n$ is a \emph{$t$ sticky-deletion with $\ell$-limited-magnitude correcting code} if for two distinct $\bx, \by \in \cC_{t,\ell}$, $D_{t,\ell}(\bx) \cap D_{t,\ell}(\by)  = \emptyset$.

\begin{lemma}(cf. \cite{levenshtein1965binary})\label{lem:opt_codesize}
An upper bound on the largest size $\cM_{sr}$ of the code $\cC_{sr}$ capable of correcting $r$ sticky insertions/deletions\footnote{$r$ denotes the maximum number of sticky-deletions, but in our definition $t$ denotes sticky-deletions occur in at most $t$ runs.} is:
\begin{equation*}
        \cM_{sr}=
    \begin{cases}
   2^{n+r}\cdot r!/n^{r}, &  {\mathrm{when}}\; r \;{\mathrm{is\; odd}}\\
   2^{n+3r/2}\cdot ((r/2)!)^2/n^{r}. & {\mathrm{when}}\; r \;{\mathrm{is\; even}}.
\end{cases}
\end{equation*}  
\end{lemma}

\begin{corollary}\label{cor:lowerbound}
    Based on Lemma~\ref{lem:opt_codesize}, a lower bound on the redundancy of the code capable of correcting $r$ sticky insertions/deletions is $r\log n+O(1)$ for constant $r$.
\end{corollary}

According to our definition of $t$ sticky-deletions with $\ell$-limited-magnitude, we can see the maximal total number of deletions is at most $t\ell$. Thus, the code $\cC_{sr}$ in Lemma~\ref{lem:opt_codesize} can trivially correct $t$ sticky-deletions with $\ell$-limited-magnitude by letting $r=t\ell$. When $r=t\ell$, based on Corollary~\ref{cor:lowerbound}, the lower bound of redundancy of the code $\cC_{sr}$ is $t\ell\log n+O(1)$ without the constraint of $\ell$-limited-magnitude. In this paper, by introducing the constraint of the magnitude of deletions, we can further reduce the code redundancy from at least $t\ell\log n+O(1)$ to at most $2t(1-1/p)\log(n+1)+O(1)$ when $\ell>2$, where $p$ is the smallest prime larger than $\ell+1$.

\begin{proposition}\label{clm:error_type}
    $t$ sticky-deletions with $\ell$-limited-magnitude in $\bx$ are equivalent to $t$ asymmetric $\ell$-limited-magnitude errors in $\bphi(\bx1)$.
\end{proposition}
\begin{proof}
    It can be easily shown by extending the aforementioned Proposition~\ref{clm:asticky}.
\end{proof}
Based on Proposition~\ref{clm:error_type}, we can see that the construction of correcting $t$ sticky-deletions with $\ell$-limited-magnitude in $\bx\in\Sigma_2^n$ is equivalent to the code construction for correcting $t$ asymmetric $\ell$-limited-magnitude errors in $\bphi(\bx1)\in\Sigma^{n_r}$, where $n_r=n_r(\bx1)$ and $\wt(\bphi(\bx1))=n+1-n_r$.

\section{Non-systematic Code Construction}\label{sec:nonsystem}
In this section, we will provide a non-systematic construction for the code capable of correcting $t$ sticky deletions with $\ell$-limited-magnitude. Then, we present the decoding algorithm of this code and a lower bound of the code size.

\begin{construction}\label{con:cons1}
	The code $\cC_{t,\ell}$ is defined as
	\begin{multline*}
		\cC_{t,\ell}=\{\bx\in\Sigma_2^n:\bphi(\bx1)\bmod q\in\cC_{q}, \\ \wt(\bphi(\bx1))=n+1-n_r\},
	\end{multline*}
 where $n_r=n_r(\bx1)$ and $\cC_{q}$ is a code over $\Sigma_{q}$ with $q=\ell+1$.
\end{construction}

\begin{lemma}\label{lem.code.1}
	$\cC_{t,\ell}$ is capable of correcting $t$ sticky-deletions with $\ell$-limited-magnitude for $\bx\in \cC_{t,\ell}$ if $\cC_{q}$ is capable of correcting $t$ symmetric errors for $\bphi(\bx1)$.
\end{lemma}
\begin{proof}
    From Proposition~\ref{clm:asticky} and \ref{clm:error_type}, we have correcting $t$ sticky-deletions with $\ell$-limited-magnitude for $\bx$ is equivalent to correcting $t$ asymmetric $\ell$-limited-magnitude errors in $\bphi(\bx1)$. Further, Theorem~4 in \cite{cassuto2010codes} has shown that the code $\cC$ is capable of correcting $t$ asymmetric $\ell$-limited-magnitude errors if $\cC_{q}$ can correct $t$ symmetric errors, where $\cC=\{\bx\in\Sigma_2^n:\bx \bmod q\in \cC_{q}\}$ and $q=\ell+1$.
\end{proof}

\begin{lemma}(\cite{aly2007quantum}, Theorem~10 )\label{lem:prime_bch}
    Let $p$ be a prime such that the distance $2\leq d\leq p^{\lceil m/2\rceil-1}$ and $n=p^m-1$. Then, there exists a narrow-sense $[n,k,d]$-BCH code $\cC_{p}$ over $\Sigma_{p}$ with
    \begin{equation*}
        n-k=\lceil (d-1)(1-1/p)\rceil m.
    \end{equation*}
\end{lemma}

\begin{theorem}\label{cor:bch}
    Let $p$ be the smallest prime such that $p\ge \ell+1$ and $n_r=n_r(\bx1)$. Then, the code $\cC_{t,\ell}$ such that
    \begin{multline*}
        \cC_{t,\ell}=\{\bx\in\Sigma_2^n:\bphi(\bx1)\bmod p\in\cC_{p}, \\ \wt(\bphi(\bx1))=n+1-n_r\}.
    \end{multline*}
    is capable of correcting $t$ sticky-deletions with $\ell$-limited-magnitude.
\end{theorem}
\begin{proof}

Let $\bx\in\cC_{t,\ell}$ be a codeword, and $\bx'\in D_{t,\ell}(\bx)$ be the output through the channel with $t$ sticky-deletions with $\ell$-limited-magnitude. Let $\bz'=\bphi(\bx'1) \bmod p$, where $p$ is the prime such that $p\ge \ell+1$. Apply the decoding algorithm of $\cC_p$ on $\bz'$ and output $\bz^{*}$. Thus, $\bz^{*}$ is also a linear code in $\cC_{p}$ and it can be shown that $\bz^{*}=\bphi(\bx1) \bmod p$. Denote $\bepsilon=(\bz^{*}-\bz') \bmod p$, we can have
\begin{equation*}
    (\bphi(\bx1)-\bphi(\bx'1)) \bmod p=(\bz^{*}-\bz') \bmod p=\bepsilon.
\end{equation*}
Hence, the output is $\bphi(\bx1)=\bphi(\bx'1)+\bepsilon$ and then recover $\bx$ from $\bphi(\bx1)$.
\end{proof}

The detailed decoding steps are shown in Algorithm~\ref{alg:decalg1}.

\begin{algorithm}[h]\label{alg:decalg1}
\SetAlgoLined
\KwInput{Retrieved Sequence $\bx'\in D_{t,\ell}(\bx)$}
\KwOutput{Decoded sequence $\bx\in\cC_{t,\ell}$.}

\textbf{Initialization:}  Let $p$ be the smallest prime larger than $\ell+1$. Also, append $1$ at the end of $\bx'$ and get $\bphi(\bx'1)$. 

\textbf{Step 1:} $\bz'=\bphi(\bx'1) \bmod p$. Run the decoding algorithm of $\cC_{p}$ on $\bz'$ to get the output $\bz^{*}$.

\textbf{Step 2:} $\bepsilon=(\bz^{*}-\bz') \bmod p$ and $\bphi(\bx1)=\bphi(\bx'1)+\bepsilon$.

\textbf{Step 3:} Output $\bx1=\bphi^{-1}(\bphi(\bx1))$ and then $\bx$.

 \caption{Decoding Algorithm of $\cC_{t,\ell}$}
 \end{algorithm}

\begin{example}
Suppose $\bx=0100111001$ and $\bx'=D_{3,2}(\bx)=010101$, where $\ell\leq 2, t=3$. Since the retrieved sequence $\bx'=010101$, then $\bphi(\bx'1)=000001$ and $\bz'=\bphi(\bx')\bmod 3=000001$, where $p$ is smallest prime such that $p\ge l+1=3$.

Run the decoding algorithm of $\cC_{p}$ on $\bz'\in\cC_{p}$, we have the output sequence $\bz^{*}=001211$ and hence $\bepsilon=(\bz^{*}-\bz)\bmod 3=001210$. Thus, the output of the decoding algorithm $\bphi(\bx1)=\bphi(\bx'1)+\bepsilon=000001+001210=001211$ and $\bx=0100111001$.
\end{example}




Next, we will present a lower bound of the size of $\cC_{t,\ell}$. The proof is given in Appendix~\ref{app:sizenonsystem}.
\begin{restatable}{theorem}{sizenonsystem}\label{thm:sizenonsystem}
	The size of the code $\cC_{t,\ell}$ in Theorem~\ref{cor:bch} is bounded by the following, where $p$ is the smallest prime larger than $\ell+1$.
    \begin{align*}
        |\cC_{t,\ell}|&\ge \frac{2^n}{p(n+1)^{2t(1-1/p)}}\\ &\ge \frac{2^n}{(2\ell+2)(n+1)^{t(2\ell+1)/(\ell+1)}}.
    \end{align*}
\end{restatable}

\begin{corollary}\label{cor:non_codesize}
    There exists a code $\cC_{t,\ell}$ capable of correcting $t$-sticky-deletions with $\ell$-limited-magnitude with  redundancy at most 
    $t(2\ell+1)/(\ell+1)\log (n+1)+O(1)$ bits.
\end{corollary}

From the above Corollary~\ref{cor:non_codesize}, we can easily notice the redundancy of this code is $\frac{3}{2}\log (n+1)+O(1)$ when $t=1,\ell=1$. Therefore, we consider the special case $t=1$ and aim to construct some codes with lower redundancy. 

In the case $t=1$ and $\ell=1$, there is only a single sticky deletion, and we can construct a code as follows.
Denote $\cC_{H}$ be the $(2^{m}-1,2^{m}-m-1,3)$-Hamming code capable of correcting single error. We take $\cC_{H}$ as $\cC_{q}$ in Construction~\ref{con:cons1} and let $n=2^m-1$. Therefore, the size of code $\cC_{1,1}$ is bounded by
$|\cC_{1,1}|\ge \frac{2^{n-1}}{ n+1}$.

Furthermore, in the case $t=1$ and $\ell$ is a given integer, we now construct a code correcting a sticky deletion of magnitude at most $\ell$.
Let $q=\ell+1$ and $p$ be the smallest prime that $p > n_r.$
For any integers $a \in [0,p-1]$ and $b\in[0,\ell]$, let the code 
$\cC_{q,a,b}$ be
\begin{align*}
\cC_{q,a,b}=\{ \bu = (u_1,\ldots,u_{n_r}) \in \Sigma_q^{n_r}:& \sum_{i=1}^{n_r} iu_i \equiv a \bmod p \\
& \sum_{i=1}^{n_r} u_i \equiv b \bmod (\ell+1)\}.
\end{align*}

We show that $\cC_{q,a,b}$ is a code correcting a single $\ell$-limited magnitude error. The proof is given in Appendix~\ref{app:codet1}.
\begin{restatable}{lemma}{codet}\label{lem:codet1}
    The code $\cC_{q,a,b}$ constructed above can correct a single $\ell$-limited magnitude error.
\end{restatable}

For any distinguish pair $(a_1,b_1) \neq (a_2,b_2)$, two codes, $\cC_{q,a_1,b_1}$ and $\cC_{q,a_2,b_2}$, are also distinguish, that is, 
\begin{equation}\label{eq.code.1}
    \cC_{q,a_1,b_1} \cap \cC_{q,a_2,b_2} = \emptyset, \;
\cup_{a=0}^{p-1} \cup_{b=0}^{\ell} \cC_{q,a,b} = \Sigma_q^{n_r}.
\end{equation}

For each code $\cC_{q,a,b}$, we define the following code
\begin{equation*}
    \cC_{t,\ell,a,b}=\{\bx \in \Sigma_2^n: \bphi(\bx1) \bmod q \in \cC_{q,a,b}\}
\end{equation*}
From \eqref{eq.code.1}, we obtain 
\begin{equation}\label{eq.code.3}
    \cC_{t,\ell,a_1,b_1} \cap \cC_{t,\ell,a_2,b_2} = \emptyset,
\end{equation}
for any pair $(a_1,b_1) \neq (a_2,b_2),$
and
\begin{equation}\label{eq.code.4}
    \cup_{a=0}^{p-1} \cup_{b=0}^{\ell} \cC_{t,l,a,b} =\Sigma_2^n.
\end{equation}
From~\eqref{eq.code.3} and \eqref{eq.code.4}, there exists integers $a,b$ such that the code $\cC_{t,\ell,a,b}$ has size at most $\frac{2^n}{p(\ell+1)}.$

Similar to the argument in Lemma \ref{lem.code.1}, we can show that the code $\cC_{t,\ell,a,b}$ can correct a sticky-deletion of magnitude at most $\ell.$ And thus, we obtain the following result.
\begin{theorem}
    There exists a code correcting a single sticky-deletion of magnitude at most $\ell$ with at most $\log n +\log (2\ell+2)$ bits of redundancy.
\end{theorem}

    

\section{Systematic Code Construction}\label{sec:system}
In the previous section, we propose a non-systematic code $\cC_{t,
\ell}$ for correcting $t$-sticky deletions with $\ell$-limited-magnitude. In this section, we will provide the efficient encoding and decoding function based on the code $\cC_{t,\ell}$ presented in Theorem~\ref{cor:bch}.

\subsection{Efficient Encoding}

Before providing the efficient systematic encoding algorithm, we now introduce a useful lemma proposed in \cite{knuth1986efficient} for encoding balanced sequences efficiently. The balanced sequence denotes the binary sequence with an equal number of $0$s and $1$s, which will be used for distinguishing the boundary of redundancy.
\begin{lemma}(cf. \cite{knuth1986efficient})
    Given the input $\bx\in\Sigma_2^k$, let the function $\bs':\Sigma_2^k\rightarrow \Sigma_2^n$ such that $\bs'(\bx)\in\Sigma_2^n$ is a balanced sequence, where $n=k+\log k$.
\end{lemma}

\begin{definition}
    Given the input $\bx\in\Sigma_2^k$, define the function $\bs:\Sigma_2^k\rightarrow \Sigma_2^{n'}$ such that $\bs(\bx)\in\Sigma_2^{n'}$ whose first bit is $1$ and $\bs(\bx)_{[2,n']}$ is balanced sequence with $(n'-1)/2$ $0$s and $(n'-1)/2$ $1$s, where $n'=k+\log k+1$.
\end{definition}

Besides, the following lemma is used for correcting $r$ 0-deletions in a binary sequence. 
\begin{lemma}(cf. \cite{dolecek2010repetition})\label{lem:lara}
For any $\bx\in\Sigma_2^k$, there exists a labeling function $f_{0}: \Sigma_2^{k}\rightarrow\Sigma_2^{n-k}$ such that the code $(\bx,f_{0}(\bx))$ is capable of correcting $r$ $0$-deletions, where $n-k=r\log n$ when $r$ is constant.
\end{lemma}

Next, we define the mapping function from non-binary to binary.
\begin{definition}
Given the input $\bx\in\Sigma_2^k$, define the function  $\bb:\Sigma_p^k\rightarrow \Sigma_2^{n}$ such that $\bb(\bu)_{[i\cdot\lceil\log p\rceil+1,(i+1)\cdot\lceil\log p\rceil]}$ is the binary form of $u_i$, where $n=k\cdot\lceil\log p\rceil $.    
\end{definition}

Given the parameters $t$ and $\ell$ in the $t$-sticky deletions with $\ell$-limited-magnitude channel, let $p$ be the smallest prime larger than $\ell+1$ and $\cC_{p}$ in Lemma~\ref{lem:prime_bch} be the $p$-ary primitive narrow-sense $[n,k,2t+1]$-BCH codes.
\begin{definition}
Define the labeling function  as $g:\Sigma_{p}^k\rightarrow \Sigma_{p}^{n-k}$ such that $(\bx,g(\bx))$ is a $p$-ary primitive narrow-sense BCH $[n,k,2t+1]$-codes, where $n=k+\lceil2t(1-1/p)\rceil m$ and $n=p^m-1$.
\end{definition}

Then, we will begin to introduce the specific encoding procedure. Suppose the input sequence is $\bc\in\Sigma_2^k$, and we have $\bphi(\bc1)$ with length $r_{c}=n_r(\bc1)$. Then, let $\bc'=\bphi(\bc1)\bmod p\in \Sigma_{p}^{r_{c}}$ and append $\mathbf{0}^{k+1-r_c}$ at the end of $\bc'$. Hence, we denote $\bar{\bc}\in\Sigma_p^{k+1}=(\bc',\mathbf{0}^{k+1-r_c})$. 

Next, encode $\bar{\bc}$ via the labeling function $g$ and output the redundancy part $g(\bar{\bc})$. Then, we will map the redundancy part $g(\bar{\bc})$ into binary sequence $\bb(g(\bar{\bc}))$ and make $\bb(g(\bar{\bc}))$ to the balanced sequence $h_1(\bar{\bc})=\bs(\bb(g(\bar{\bc})))$.

Further, we need to protect the redundancy part $h_1(\bar{\bc})$. The idea is to apply the code in Lemma~\ref{lem:lara} on $h_1(\bar{\bc})$ and output $f_0(h_1(\bar{\bc}))$. Also, make $f_0(h_1(\bar{\bc}))$ to balanced sequence $\bs(f_0(h_1(\bar{\bc})))$ and repeat its each bit $2t\ell+1$ times. Let
$h_2(\bar{\bc})=\mathrm{Rep}_{2t\ell+1}\bs((f_0(h_1(\bar{\bc}))))$, where $\mathrm{Rep}_{k}(\bx)$ is the $k$-fold repetition of $\bx$.

Finally, we have the output $\mathrm{Enc}(\bc)=(\bc,\bpsi^{-1}( h(\bc)))$, where $h(\bc)=(h_1(\bar{\bc}),h_2(\bar{\bc}))$.  The detailed encoding steps are summarized in the following Algorithm~\ref{alg:enc}.

\begin{algorithm}[h]\label{alg:enc}
\SetAlgoLined
\KwInput{$\bc\in \Sigma_2^k$}
\KwOutput{Encoded sequence $\mathrm{Enc}(\bc)\in\Sigma_2^N$}

\textbf{Initialization:}  Let $p$ be the smallest prime larger than $\ell+1$. 

\textbf{Step 1:} Append $1$ at the end of $\bc$ and get $\bphi(\bc1)$ with length $r_c=n_r(\bc1)$.

\textbf{Step 2:} $\bc'=\bphi(\bc1)\bmod p\in\Sigma_p^{r_c}$. Append $\mathbf{0}^{k+1-r_c}$ at the end of $\bc'$, then $\bar{\bc}=(\bc',\mathbf{0}^{k+1-r_c})$.

\textbf{Step 3:} Encode $\bar{\bc}$ via $\cC_{p}$ and output $g(\bar{\bc})$. Mapping 
$g(\bar{\bc})$ to balanced binary sequence $h_1(\bar{\bc})=\bs(\bb(g(\bar{\bc})))$. 

\textbf{Step 4:} Protect $h_1(\bar{\bc})$ via $f_{0}$ and obtain the total redundancy $h(\bc)=(h_1(\bar{\bc}),h_2(\bar{\bc}))$.
    
\textbf{Step 5:} Output $\mathrm{Enc}(\bc)=(\bc,\bpsi^{-1}(h(\bc)))\in\Sigma_2^N$.
 \caption{Encoding Algorithm}
 \end{algorithm}




\begin{lemma}\label{lem:encodingdense}
Given a sequence $\bc\in\Sigma_2^{k}$, Algorithm~\ref{alg:enc} outputs an encoded sequence capable of correcting $t$ sticky-deletions with $\ell$-limited-magnitude $\mathrm{Enc}(\bc)\in\Sigma_2^N$.
\end{lemma}

Therefore, the total redundancy of the code  $h(\bc)=(h_1(\bar{\bc}),h_2(\bar{\bc}))$ via this encoding process can be shown as follows. The proof is given in Appendix~\ref{app:totalred}.
\begin{restatable}{theorem}{totalred}\label{thm:totalred}
The total redundancy of the code $\mathrm{Enc}(\bc)\in\Sigma_2^N$ by given input $\bc\in\Sigma_2^k$ is
   \begin{multline*}
    N-k= \frac{\lceil2t(1-1/p)\rceil\cdot\lceil\log p\rceil}{\log p} \log(N+1)\\+O(\log\log N).
\end{multline*} 
where $p$ is smallest prime such that $p\ge \ell+1$. 
\end{restatable}

\subsection{Decoding Algorithm}

Without loss of generality, suppose the encoded sequence $\mathrm{Enc}(\bc)\in\Sigma_2^N$ is transmitted through the $t$ sticky deletions with $\ell$-limited-magnitude channel, and we have the retrieved sequence $\bd\in\Sigma_2^{N-t\ell}$. In this subsection, we will introduce the decoding algorithm for obtaining $\mathrm{Dec}(\bd)\in\Sigma_2^{k}$ by given  $\bd\in\Sigma_2^{N-t\ell}$. We will introduce the explicit decoding procedure as follows.

First, we get $\bpsi(\bd)$, which is the derivative of $\bd$, but we need to distinguish where the redundancy part begins. Since $t$ sticky deletions with $\ell$-limited magnitude occurred in $\mathrm{Enc}(\bc)$ is equivalent to deleting at most $t\ell$ $0$s in $\bpsi(\mathrm{Enc}(\bc))$, the number of $1$s in $\bpsi(\bd)$ is the same with that of in $\bpsi(\mathrm{Enc}(\bc))$. Thus, we can count the number of $1$s from the end of $\bpsi(\bd)$ to find the beginning of the redundancy since the redundancy part is the balanced sequence.

Hence, we find the $(n_2+2t\ell+1)/2$-th $1$ and $(n_1/2+n_2/2+t\ell+1)$-th $1$ from the end of $\bpsi(\bd)$ and denote their entries as $i_{r2}$ and $i_{r1}$, respectively. For the subsequence $\bpsi(\bd)_{[i_{r2},N-t\ell]}$, since there are at most $t\ell$ $0$s are deleted in $\mathrm{Enc}(\bc)_{[N-n_2+1,N]}$, the $(2t\ell+1)$-fold repetition code can help recover $\bs(f_0(h_1(\bar{\bc})))$. Further, we can obtain parity bits $f_0(h_1(\bar{\bc}))$.

Next, for the subsequence $\bpsi(\bd)_{[i_{r1},i_{r2}-1]}$, there are also at most $t\ell$ $0$s are deleted in $\mathrm{Enc}(\bc)_{[N-n_1-n_2+1,N-n_2]}$. Thus, the code introduced in Lemma~\ref{lem:lara} can help to protect the redundancy part of $\mathrm{Enc}(\bc)_{[N-n_1-n_2+1,N-n_2]}$. Then, $h_1(\bar{\bc})$ can be recovered with the help of parity bits $f_0(h_1(\bar{\bc}))$. Further, we can get the $g(\bar{\bc})$ from $h_1(\bar{\bc})=\bs(\bb(g(\bar{\bc})))$. 

Finally, denote $\bz=(\bphi(\bd_{[1,i_{r1}-1]},1),\mathbf{0}^{k+1-r_c})\in\Sigma^{k+1}$ and $\bz'=\bz \bmod p$, where $r_c$ is the length of $\bphi(\bd_{[1,i_{r1}-1]},1)$ and $k=N-n_1-n_2$. Then, the following decoding steps are the same as Algorithm~\ref{alg:decalg1} where $\bz'$ is the input of Step 1 of Algorithm~\ref{alg:decalg1}. The only difference is we need to first remove $\mathbf{0}^{k+1-r_c}$ at the end before the last step of $\bphi^{-1}$. Therefore, the main steps for decoding $\bd\in\Sigma_2^{N-t\ell}$ is summerized in Algorithm~\ref{alg:dec2}.

\begin{algorithm}[h]\label{alg:dec2}
\SetAlgoLined
\KwInput{$\bd\in \Sigma_2^{N-t\ell}$}
\KwOutput{Decoded sequence $\mathrm{Dec}(\bd)\in\Sigma_2^k$}

\textbf{Initialization:}  Let $p$ be the smallest prime larger than $\ell+1$. 

\textbf{Step 1:} Get $\bpsi(\bd)$. Find the $(n_2+2t\ell+1)/2$-th $1$ and $(n_1/2+n_2/2+t\ell+1)$-th $1$ from the end of $\bpsi(\bd)$ and denote their entries as $i_{r2}$ and $i_{r1}$, respectively.

\textbf{Step 2:} Recover $\bs(f_0(h_1(\bar{\bc})))$ from $\bpsi(\bd)_{[i_{r2},N-t\ell]}$ and then get $f_0(h_1(\bar{\bc}))$.

\textbf{Step 3:} Recover $h_1(\bar{\bc})$ via $f_0(h_1(\bar{\bc}))$ and then obtain $h_1(\bar{\bc})$.

\textbf{Step 4:} Denote $\bz'=(\bphi(\bd_{[1,i_{r1}-1]},1),\mathbf{0}^{k+1-r_c})\bmod p$. Input $\bz'$ to Step 1 of Algorithm~\ref{alg:decalg1} and run the remaining steps of Algorithm~\ref{alg:decalg1}.
    
\textbf{Step 5:} Output $\mathrm{Dec}(\bd)$.

\caption{Decoding Algorithm}
\end{algorithm}

\subsection{Time Complexity}

For the encoding algorithm, given  constants $t,\ell$, each codeword is generated by following steps:
\begin{itemize}
\item First, given an input binary message string $\bc$ and output $\bar{\bc}$. The time complexity is $O(n)$.
\item Second, encode $\bar{\bc}$ via $\cC_{p}$. The time complexity of the $p$-ary narrow-sense $[n,k,2t+1]$-BCH code $\cC_{p}$ is $O(tn\log n)$.
\item Third, map the labeling function of $\bar{\bc}$ to balanced binary sequence with the time complexity $O(\log n)$.
\item Forth, protect the redundancy part $h_1(\bar{\bc})$ via the code in Lemma~\ref{lem:lara} and make it to the balanced sequence and the repetition code with the time complexity $O(\log n)$.
\end{itemize}
Therefore, the time complexity of the encoder time complexity is dominated by the $p$-ary narrow-sense BCH code, which is $O(tn\log n)$.

For the decoding algorithm, we can easily show that the time complexity is dominated by the decoding of the $p$-ary narrow-sense BCH code and decoding for the code in Lemma~\ref{lem:lara}. Different from the encoding procedure, the decoding for the code in Lemma~\ref{lem:lara} is brute-force, hence the time complexity is $O((n''_{2})^{t\ell})=O((\log n)^{t\ell})$. Therefore, the total time complexity of decoding is $O(tn+(\log n)^{t\ell})$.

\section{conclusion}\label{sec:conclu}
In this paper, we presented codes for correcting $t$-sticky deletions with $\ell$-limited-magnitude. We first presented a non-systematic code for this type of error and analyze its size. We then developed systematic codes and proposed efficient encoding and decoding algorithms. However, there still remain some interesting problems, including extending this work to a larger number of deletions, not only constant $t$.



\bibliographystyle{IEEEtran}
\bibliography{reference}

\begin{appendices}

\section{Proof of Lemma~\ref{lem:size_phix}}\label{app:size_phix}
\sizephi*
\begin{proof}
    For a binary sequence $\bx\in\Sigma_2^n$, the corresponding sequence $\bphi(\bx1)$ is with length $n_r=n_r(\bx1)$ and $\wt(\bphi(\bx1))=n+1-n_r$. Also, the cardinality of $\Phi$ can be considered the number of ways of arranging $n+1-n_r$ indistinguishable objects in $n_r$ distinguishable boxes. Thus, we can get the cardinality of $\Phi$ as shown in Lemma~\ref{lem:size_phix}. 

On the other side, since the mapping function $\bphi$ is a one-to-one mapping function, the cardinality of $\Phi$ should be the same as $|\Sigma_2^n|=2^n$. 
\end{proof}

\section{Proof of Theorem~\ref{thm:sizenonsystem}}\label{app:sizenonsystem}
\sizenonsystem*

\begin{proof}
    Denote $\bz=\bphi(\bx1)\bmod p$. $\bphi(\bx1)$ can be written as $\bphi(\bx1)\rightarrow(\bz,\ba)$ such that $\bphi(\bx1)=\bz+p\cdot \ba$, where $\ba$ is a vector with the same length as $\bphi(\bx1)$ and $\bz$. 
    Further, since $\bz\in\cC_{p}$ and $\cC_{p}$ is a linear code, the code $\cC_{p}$ with length $n_r$ can be considered as a set which is obtained by $\Sigma_{p}^{n_r}$ partitioned into $p^{n_r-k}$ classes.
    
    Denote $\bphi(\bx1)^{n_r}$ as the $\bphi(\bx1)$ with length $n_r$. Thus, for any fixed number of runs $n_r$, the cardinality of $\bphi(\bx1)^{n_r}$ such that $\bphi(\bx1)^{n_r} \bmod p \in\cC_{p}$ with length $n_r$ is:
    \begin{equation*}
    \left|\bphi(\bx1)^{n_r}\right|=\frac{{n \choose n_r-1}}{p^{n_r-k}}.
    \end{equation*}
    Then, the size of the code $\cC_{t,\ell}$ in Theorem~\ref{cor:bch} can be shown as:
    
    \begin{align}\label{eq:size_c}
|\cC_{t,\ell}|=\sum_{n_r=1}^{n+1}\left|\bphi(\bx1)^{n_r}\right|&=\sum_{n_r=1}^{n+1}\left[\frac{{n \choose n_r-1}}{p^{n_r-k}}\right]\nonumber\\
&\ge \frac{\sum_{n_r=1}^{n+1}{n \choose n_r-1}}{p^{n+1-k}}=\frac{2^n}{p^{n+1-k}}.
\end{align} 
From Lemma~\ref{lem:prime_bch} and Theorem~\ref{cor:bch}, let $d=2t+1$ and $m=\log_p(n+1)$.
\begin{align}\label{eq:p_nk}
    p^{n+1-k}&=p^{2t(1-1/p)\cdot \log_{p}(n+1)+1}=p(n+1)^{2t(1-1/p)}\nonumber\\ 
    &\leq (2\ell+2)(n+1)^{t(2\ell+1)/(\ell+1)}.
\end{align}
where the last inequality from the fact that $\ell+1\leq p<2(\ell+1)$ based on the well-known Bertrand–-Chebyshev theorem.

Therefore, from \eqref{eq:size_c} and \eqref{eq:p_nk}, the size of the code $\cC_{t,\ell}$ in Theorem~\ref{cor:bch} is bounded by
    \begin{equation*}
        |\cC_{t,\ell}|\ge \frac{2^n}{(2\ell+2)(n+1)^{t(2\ell+1)/(\ell+1)}}.\qedhere
    \end{equation*}
\end{proof}

\section{Proof of Lemma~\ref{lem:codet1}}\label{app:codet1}
\codet*
\begin{proof}
Let $\bu =(u_1,\ldots,u_{n_r})$ be an original code word and $\bv =(v_1,\ldots,v_{n_r})$ be a word obtained from $\bu$ with at most a single $\ell$-limited magnitude error. That is, $\sum_{i=1}^{n_r} u_i - \sum_{i=1}^{n_r} v_i = t \leq \ell$. Hence, if $t \neq 0,$ there is an index $i_0$ such that $u_{i_0}-v_{i_0}=t.$ And thus, $\sum_{i=1}^{n_r} v_i + i_0 t =\sum_{i=1}^{n_r} u_i \equiv a \bmod p.$
We now show that the index $i_0$ that satisfies the above condition is unique. Assume that there are two indices $i_1$ and $i_2$ such that  $\sum_{i=1}^{n_r} v_i + i_1 t = \sum_{i=1}^{n_r} v_i + i_2 t \equiv a \bmod p$. Then, $(i_2-i_1)t \equiv 0 \bmod p.$ It is not possible since $p$ is a prime and both $(i_2-i_1), t < p.$
Therefore, we can determine the index $i_0$ uniquely and recover the original word $\bu.$ So, the code $\cC_{q,a,b}$ can correct a single $\ell$-limited magnitude error.
\end{proof}

\section{Proof of Theorem~\ref{thm:totalred}}\label{app:totalred}
\totalred*
\begin{proof}
Let $m=\log_p(N+1)$, hence $N=p^m-1$.  The lengths of the redundancy parts are as follows:
\begin{itemize}
    \item $n''_1$ is the length of $g(\bar{\bc})$: $n''_1=\lceil2t(1-1/p)\rceil m$;
    \item $n'_1$ is the length of $\bb(g(\bar{\bc}))$: $n'_1=n''_1\cdot\lceil\log p\rceil$;
    \item $n_1$ is the length of $\bs(\bb(g(\bar{\bc})))$: $n_1=n'_1+\log n'_1+1$;
    \item $n''_2$ is the length of $f_0(h_1(\bar{\bc}))$: $n''_2=t\ell\log n_1$;
    \item $n'_2$ is the length of $\bs(f_0(h_1(\bar{\bc})))$: $n'_2=n''_2+\log n''_2+1$;
    \item $n_2$ is the length of $h_2(\bar{\bc})$: $n_2=(2t\ell+1)n'_2$;
    
\end{itemize}
    Based on the above statement, we can see that $N-k=n_1+n_2$, where
    \begin{equation*}
        n'_1=(\lceil2t(1-1/p)\rceil m)\cdot\lceil\log p\rceil
    \end{equation*}
with $m=\log_{p}(N+1)$. Hence, we have
\begin{equation*}
        n'_1=\frac{\lceil2t(1-1/p)\rceil\cdot\lceil\log p\rceil}{\log p} \log(N+1)
    \end{equation*}
Since both $t$ and $p$ are constants, then $\log n'_1=O(\log\log N)$ and $n_2=O(\log\log N)$.
Therefore, the total redundancy of the code $\mathrm{Enc}(\bc)\in\Sigma_2^N$ given the input $\bc\in\Sigma_2^k$ can be shown as the Theorem~\ref{thm:totalred}.
\end{proof}
\end{appendices}

\end{document}